\documentclass[copyright,creativecommons]{eptcs}

\usepackage{booktabs}
\usepackage{iftex}

\ifpdf
  \usepackage{underscore}         
  \usepackage[T1]{fontenc}        
\else
  \usepackage{breakurl}           
\fi

\title{A Categorical Unification for Multi-Model Data: Part II Categorical Algebra and Calculus \footnote{To be published in the Proceedings of the 2025 Applied Category Theory (ACT) Conference.} }
\author{Jiaheng Lu
\institute{University of Helsinki, Finland}
\email{jiaheng.lu@helsinki.fi}
}
\def\titlerunning{}
\def\authorrunning{Jiaheng Lu}

\usepackage{tikz-cd}
\usepackage[linesnumbered,ruled,vlined]{algorithm2e} 
\usepackage{enumitem}
\usepackage{multirow}
\usepackage[yyyymmdd,hhmmss]{datetime}
\usepackage{tcolorbox}
\usepackage{mathtools}
\usepackage{amssymb}
\usepackage{amsmath} 
\usepackage{amsthm}

\newtheorem{theorem}{Theorem}
\newtheorem{lemma}[theorem]{Lemma}

\theoremstyle{definition}
\newtheorem{definition}[theorem]{Definition}
\newtheorem{example}[theorem]{Example}

\theoremstyle{remark}

\newcommand{\myeq}{\overset{\mathrm{def}}{=}}

\begin{document}
\maketitle

\begin{abstract}
Multi-model databases are designed to store, manage, and query data in various models, such as relational, hierarchical, and graph data, simultaneously.  In this paper, we provide a theoretical basis for querying categorical databases. We propose two formal query languages: categorical calculus and categorical algebra, by extending relational calculus and relational algebra respectively. We demonstrate the equivalence between these two languages of queries. We propose a series of transformation rules of categorical algebra to facilitate query optimization. Finally, we analyze the expressive power and computation complexity for the proposed query languages.

\end{abstract}

\section{Introduction}

The ``Variety” of data is one of the most important issues in modern data management systems \cite{Lu:2019:MDN:3341324.3323214,journals/pvldb/KiehnSGPWWPR22}.  Across various applications, data sources inherently exhibit diverse organizational structures and formats, e.g. relation, graph, XML, JSON, etc.  In order to provide a unified view and query interface, a categorical model (e.g. \cite{Brown2019CategoricalDI,lu2025categoricalunificationmultimodeldata,DBLP:journals/corr/abs-2201-04905}) has been proposed to define a formal, unified data model capable of effectively accommodating this diverse array of data models. 

This paper propose a comprehensive query framework within the categorical paradigm. We establish a theoretical groundwork for a multi-model query language through the introduction of categorical algebra and categorical calculus. Recall that relational calculus and relational algebra are two theoretical languages used for expressing queries on relational databases. They provide a way to formulate queries that can retrieve and manipulate data stored in relational database systems. Analogously, this paper develops categorical calculus and  algebra for categorical databases.  While categorical calculus is a declarative language that describes the properties of desired objects and morphisms in categories, categorical algebra is a procedural language that provides operations to manipulate categories and retrieve specific objects and morphisms from them.

The contributions and organization of the paper are as follows:

(i) We introduce categorical calculus for formulating queries on categorical databases. This calculus extends the relational domain calculus by incorporating multi-model predicates. These predicates enable the formulation of queries such as reachability queries in graph data, and twig pattern matching for parent-child and ancestor-descendant relationships in XML data (Section \ref{sec:calculus}).

(ii) We propose categorical algebra that includes both set operations and category operations. Set operations involve manipulating individual or multiple sets, while category operations enable the conversion between sets, functions and categories. Categorical algebra can address a wide range of queries, including relational queries, XML twig pattern matching, graph pattern matching, and reachability queries on graphs. We establish a theorem demonstrating the equivalence between categorical calculus and categorical algebra. Specifically, each categorical calculus can be expressed using a set of operators of categorical algebra, and vice versa (Section \ref{sec:algebra}).

(iii) We define a series of transformation rules of algebra expressions, which rewrite categorical algebraic expressions into more efficient forms for query optimization. We analyze the expressive power and computation complexity for the proposed query languages (Section \ref{sec:rules}).


\noindent \textbf{Related work.}  This work is related to two categories of previous research: (1) applied category theory for query languages in databases, and (2) query calculus and algebra for various types of data.

\smallskip


Although category theory was initially developed as an abstract mathematical language to unify algebra and topology, researchers have explored its applications in various fields, including database query languages. Tannen (1994) \cite{conf/pods/Tannen94} highlights that concepts from category theory can lead to a series of useful languages for collection types such as sets, bags, lists, complex objects, nested relations, arrays, and certain kinds of trees.  Libkin (1995) \cite{DBLP:conf/icdt/Libkin95} uses universality properties to study the syntax of query languages with approximations. Libkin and Wong (1997) \cite{journals/jcss/LibkinW97} showcase the connection between database operations on collections and the categorical notion of a monad. More recently, Schultz and Spivak (2016) \cite{schultz2016algebraic} introduce a categorical query language that serves as a data integration scripting language. Gibbons et al. (2018) \cite{journals/pacmpl/GibbonsHHW18} present an elegant mathematical foundation for database query language design using category theory concepts such as monads and adjunctions. These works demonstrate the intriguing connections between category theory and query languages and query processing, revealing how category theory can inform and enhance database query language design.

\smallskip


Since Codd (1972)  \cite{DBLP:persons/Codd72} demonstrated that the same class of database queries can be expressed by relational algebra and the safe formulas of relational calculus, query algebra and calculus have become central topics in database research. A plethora of related works exist on query algebra and calculus across various types of data, including 
relational data \cite{DBLP:persons/Codd72,10.1145/114325.103712,10.1145/99935.99943,10.1145/32204.32219}, nested relational data \cite{DBLP:conf/pods/ParedaensG88,DBLP:conf/pods/Gucht87}, graph data \cite{DBLP:conf/pods/FrancisGGLMMMPR23,10.1093/comjnl/bxaa031}, object-oriented data \cite{10.1145/588011.588026}, incomplete and probabilistic data \cite{suciu:OASIcs.Tannen.10,10.1145/1265530.1265535}, etc. These previous papers lay the foundation for database systems to process various types of data. Our paper makes contributions by proposing query calculus and algebra for categorical data with the application on multi-model databases.


\section{Preliminaries}

\begin{figure}\centering\includegraphics[width=0.85\textwidth]{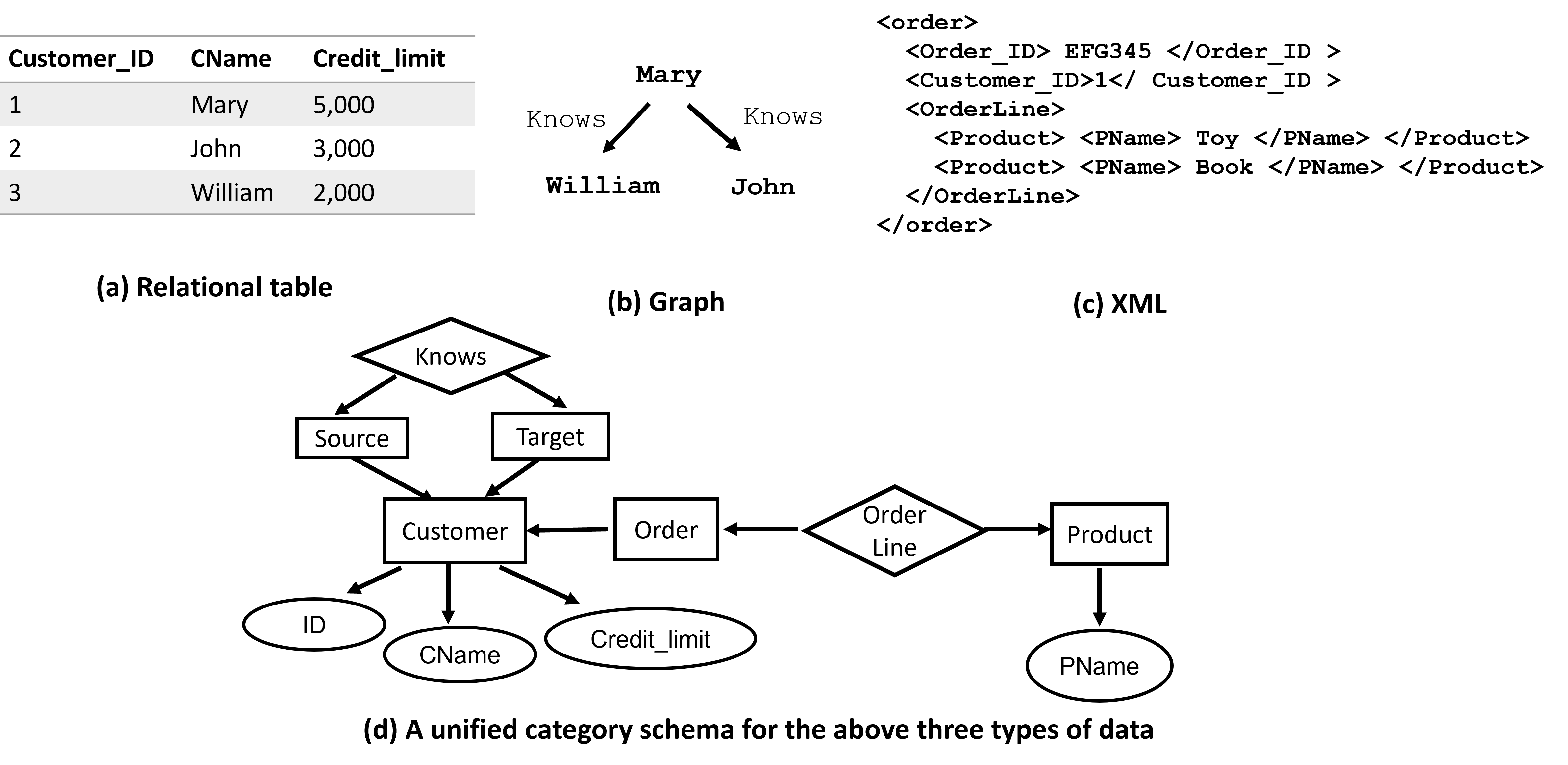}
\caption{This toy example shows three types of data including relation, XML and graph. They have a unified categorical representation.} \label{fig:firstexample}
\end{figure}


A database can be viewed as a category, specifically a set category. In this context, each object in the category represents a set, while each morphism corresponds to a function between two sets.  The composition of morphisms aligns with the composition of functions. In the rest of this paper, we will use the terms ``\textit{set}” and ``\textit{object}” interchangeably, as well as ``\textit{function}” and ``\textit{morphism}” interchangeably. There are three fundamental types of objects within this category: \textit{entity objects}, \textit{attribute objects}, and \textit{relationship objects}. For example, Figure \ref{fig:firstexample} illustrates a unified category representation of three types of data (relation, XML, and graph) from E-commerce, which contains customers, social network, and order information with three distinct data models. Customer information is stored in a relational table—their IDs, names, and credit limits. Graph data bear information about mutual relationships between the customers, i.e., who knows whom. In XML documents, each order has an ID and a sequence of ordered items, each of which includes product name. Those three types of data can be integrated in one unified category schema in Figure \ref{fig:firstexample} (d), where ``\texttt{Knows}'' is a relationship object, ``\texttt{Customer}'', ``\texttt{Order}'', and ``\texttt{Product}'' are entity objects and ``\texttt{ID}'', ``\texttt{CName}'', ``\texttt{PName}'' are attribute objects. Furthermore, the category under discussion is not just any set category, but rather a thin category,  defined as follows:




\begin{definition}(\textbf{Thin Category or Posetal Category}) \cite{roman2017introduction} Given a pair of objects $X$ and $Y$ in a category $\mathcal{C}$, and any two morphisms $f$, $g$: $X \to Y$, we say that $\mathcal{C}$ is a thin category (or posetal category) if and only if the morphisms $f$ and $g$ are equal. 
\label{def:thin}\end{definition}

In category theory, a commutative diagram is a graphical representation that depicts the relationships between objects and morphisms.  A diagram is commutative meaning that different paths through the diagram yield the same result. All diagrams in a thin set category are commutative \cite{roman2017introduction}.





\section{Categorical Calculus}
\label{sec:calculus}

A query on categories is a mapping from one category $\mathcal{C}$ into another category $\mathcal{D}$. In the following two sections, we will define two formal query languages: (1) categorical calculus, a declarative language for describing results in $\mathcal{D}$; and (2)  categorical algebra, a procedural language for listing operations on $\mathcal{C}$.







The alphabets for categorical calculus are listed in Table \ref{tab:notations}. Given a category $\mathcal{C}$, let $S_i$ denote an object in $\mathcal{C}$, and $f:S_i \to S_j$ be a function between $S_i$ and $S_j$. Given any element $x \in S_i$, there are two ways to express the image $y$ of $x$ under $f$: (i) $y = f(x)$ and (ii)  $y = x \cdot S_j $. The second expression is unambiguous because there is only one function between any two objects in a thin category (Refer to Def  \ref{def:thin}). Furthermore, to present the composition of functions, consider two functions  $f:S_i \to S_j$ and $g:S_j \to S_k$. Given any element $x \in S_i$,  the image $y$ of $x$ under the composition of $f$ and $g$ can be written as $y= (g \circ f)(x)$ or $y= x \cdot S_j \cdot S_k$. For example, consider Figure \ref{fig:firstexample}. Assume $x$ is an element of \texttt{OrderLine}, then the product name $y$ of $x$ is denoted as $y = x \cdot$ \texttt{Product} $\cdot$ \texttt{Name}.

 We then define three types of terms: range terms, function terms, and predicate terms.

\textbf{Range terms}: A range term $x \in S$ specifies that the element $x$ belongs to the set $S$.



\textbf{Function terms}: Given two  variables $x_i \in S$ and $x_j \in S'$, a function term $(f: x_i \to x_j) = g_n \circ \cdots  g_2 \circ g_1 $ shows that $f$ is a function from $S$ to $S'$ and $f$ is the composition of a sequence of functions $g_1$ , $g_2$, ..., and $g_n$.




\textbf{Predicate terms}: Let \( x_i \) and \( x_j \) be variables, \( \alpha \) an individual constant, and \( \theta \) a predicate symbol. The expressions \( x_i \theta x_j \) and \( x_i \theta \alpha \) are referred to as predicate terms. Predicates can be classified into three types:

1. \textbf{Classic Predicates (\( \theta_M \))}: These include standard mathematical comparison symbols such as \( =, \neq, <, >, \le, \ge \), which are used to compare numeric or character data.

2. \textbf{Tree Data Predicates (\( \theta_T \))}: These predicates are specifically designed for operations on tree data structures. Consider a category describing tree data $\mathcal{T}$, where each tree node is assigned a Dewey code \cite{conf/sigmod/TatarinovVBSSZ02,DBLP:conf/vldb/LuLCC05} to indicate its position within $\mathcal{T}$. In Dewey codes, each node is represented by a vector: (i) the root is labeled with an empty string $\epsilon$; (ii) for a non-root node $u$, label($u$) = label($s$).$x$, where $u$ is the $x$-th child of $s$. Dewey codes facilitate efficient evaluation of structural relationships between elements. For example, if a node $u$ is labeled "1.2.3.4", then its parent is labeled "1.2.3" and its grandparent is labeled "1.2".

The predicate $\theta_T$ can include \texttt{isParent}, \texttt{isChild}, \texttt{isAncestor}, \texttt{isDescendant}, \texttt{isFollowing}, \texttt{isSibling}, \texttt{isFollowing-sibling}, \texttt{isPreceding}, and \texttt{isPreceding-sibling}. All these XPath axes \cite{10.1145/2463664.2463675} can be determined by comparing the two Dewey codes.


3. \textbf{Graph Data Predicates (\( \theta_G \))}: These predicates are tailored for operations on graph data structures. Consider a category representing a graph  $\mathcal{G}$, where all edges are stored in a set $E$. In this context, if $a$ and $b$ are two node variables in $\mathcal{G}$, then the notation $a \rightsquigarrow^E b$ indicates that node $b$ is reachable from node $a$ via edges in $E$. Additionally, $a \rightsquigarrow_n^E b$ shows that node $b$ can be reached from node $a$ within $n$ hops.





\begin{table}
\caption{The Alphabets of the Categorical Calculus}
\label{tab:notations}
\begin{center}
\begin{tabular}{ |c|c| } 
 \hline
Notation & Examples  \\ [0.5ex] 
  \hline  \hline 
   Object Constants & $S_1,S_2 \cdots$   \\ 
 \hline 
Morphism Constants & $f_1 : S_1 \to S_2 = g_1 \circ g_2 \cdots$  \\  
  \hline
  Element Variables in Objects & $x_1 \in S_1, x_2 \in S_2 \cup S_3 \cdots$   \\ 
 \hline 
 Logical Symbols & $\exists,\forall,\wedge,\vee,\neg$   \\ 
 \hline 
\end{tabular}
\end{center}
\end{table}

\smallskip

 A (well-formed) formula of the categorical calculus is defined recursively as follows:

\begin{enumerate}[noitemsep]
  \item Any term (including range term, function term and predicate term) is a formula;
  \item If $\alpha$ is a formula, so is $\neg \alpha$;
  \item If $\alpha_1, \alpha_2$ are formulas, so are ($\alpha_1 \vee \alpha_2$) and ($\alpha_1 \wedge \alpha_2$);
  \item If $\alpha$ is a formula in which $x$ occurs as a free variable, then $\exists x (\alpha)$ and $\forall x (\alpha)$  are formulas.
\end{enumerate}


Table \ref{tab:examples} provides several examples of formulae. The standard definitions of \textit{free} and \textit{bound} occurrences of variables are utilized here.  A free variable in a formula is one that is not quantified by any quantifier within the formula; it is not bound by logical operators such as $\forall$ or $\exists$. In contrast, a bound variable is one that is quantified by a quantifier within the formula.

An expression in categorical calculus typically takes the form:
\[ Z = \{ X, \mathcal{R}  \mid W \} \]

where (1)  $X = {x_1, \ldots, x_m}$ are $m$ distinct variables representing entity or attribute objects;

(2) $\mathcal{R} = {R_1, \ldots, R_n}$, where each $R_i = (x_{i_1}, \ldots, x_{i_j})$ represents a relational object with $j$ components. Either $X$ or $\mathcal{R}$ may be empty, but not both;

(3) $W = U_1 \wedge U_2 \wedge ... \wedge U_p \wedge F_1 \wedge F_2 \wedge ... \wedge F_{k}  \wedge V $, where 

(3.1) $U_1$ through $U_p$ are range terms over variables in $X$ and $\mathcal{R}$;

(3.2) $F_1$ through $F_{k}$ are function terms over objects in $X$ and $\mathcal{R}$; 

(3.3)  $V$ is either null, or it is a formula with the properties: (3.3.1) Every bound variable in $V$ has a clearly defined range, ensuring they are \textit{safe} expressions, which will be explained below.
(3.3.2) Every free variable in $V$ belongs to the variable in $X$ and $\mathcal{R}$.

\begin{table}
\caption{The formulae of the Categorical Calculus}
\label{tab:examples}
\begin{center}
\begin{tabular}{ |c|c| } 
 \hline

 Formulae with range terms  & Safe variables   \\ 
 \hline  \hline
 $ x_1 \in O_1$ & $x_1$   \\ 
 \hline 
 $ x_1 \in O_1 \wedge \neg (x_1 \in O_2) $ & $x_1$   \\ 
 \hline  \hline
Formulae with function and range terms  & Safe variables   \\ 
 \hline   \hline
 $ ((f_1: x_1 \to x_2) = f_2 \circ g_1)  \wedge  (x_1 \in S_1) \wedge (x_2 \in S_2)$ & $x_1,x_2$ \\
 \hline   
 $ (\pi_1: (x_1,x_2) \to x_1 )  \wedge  (x_1 \in S_1) \wedge (x_2 \in S_2)$ & $x_1,x_2$   
 \\ 
 \hline \hline
 Formulae with predicate, range and function terms & Data model   \\ 
 \hline  \hline
 $    (x_1 \in S_1) \wedge  (x_2 \in S_2) \wedge (x_1 \rightsquigarrow^E x_2)$ $\wedge (x_1 \cdot$ Name = "John") & Graph   \\ 

 \hline
 $   (x_1 \in D_1) \wedge  (x_2 \in D_2) \wedge (x_1 ~ isAncestor ~ x_2)$ &  Tree  \\ 
 \hline \hline
 Formulae with unsafe terms &  Unsafe variable  \\ 
  \hline  \hline 

 $x_2 \in S_2, x_3 \in S_3, \exists x_1 ( x_1 > x_3 \wedge x_2 = 6)$  & $x_1$   \\
 \hline

 $\forall x_1 \exists x_2 \in S_2 ( x_1 > x_2$)  & $x_1$   \\ 
 \hline  
 $   (x_1 \in S_1) \vee f(x_1)=a_1$ & $x_1$  \\ 
 \hline 
\end{tabular}

\end{center}

\end{table}

Whenever we use universal quantifiers, existential quantifiers, or negative of predicates in a categorical calculus expression, we must make sure that the resulting expression make sense. A safe expression is one that is guaranteed to yield a finite number of elements in any object as its result; otherwise the expression is called unsafe. That is, in a safe expression, both bound and free variables must have clearly defined their ranges.

 \begin{example} \label{exp:calculus} Consider the categorical schema depicted in Figure \ref{fig:SC2}(a), which stores information about students, their addresses, their gender, and the courses they attend (in the \texttt{SC} object). Let's examine a query scenario: \textit{finding all male students and their addresses who attend every course attended by at least one female student}. Figure \ref{fig:SC2}(b) and (c) illustrate the structure of the query and the returned category  respectively. The categorical calculus is formulated as follows:

\begin{gather*}
\left\{ (x_1, x_2) \mid x_1 \in \text{Student}, x_2 \in \text{Address}, 
f(x_1) = x_2 = f_2, x_1\cdot\text{Gender} = \text{'Male'}, \right. \\
\left. \forall y_1 \in \text{Student}, y_1\cdot\text{Gender} = \text{'Female'} \to 
\exists y_2, y_3 \in \text{SC}, \exists y_4 \in \text{Course}, \right. \\
\left.
(y_2 \cdot \text{Student} = x_1) \wedge (y_3 \cdot \text{Student} = y_1) \wedge 
(y_2 \cdot \text{Course} = y_4) \wedge (y_3 \cdot \text{Course} = y_4)
\right\}
\end{gather*}




 \begin{figure}\centering\includegraphics[width=0.9\textwidth]{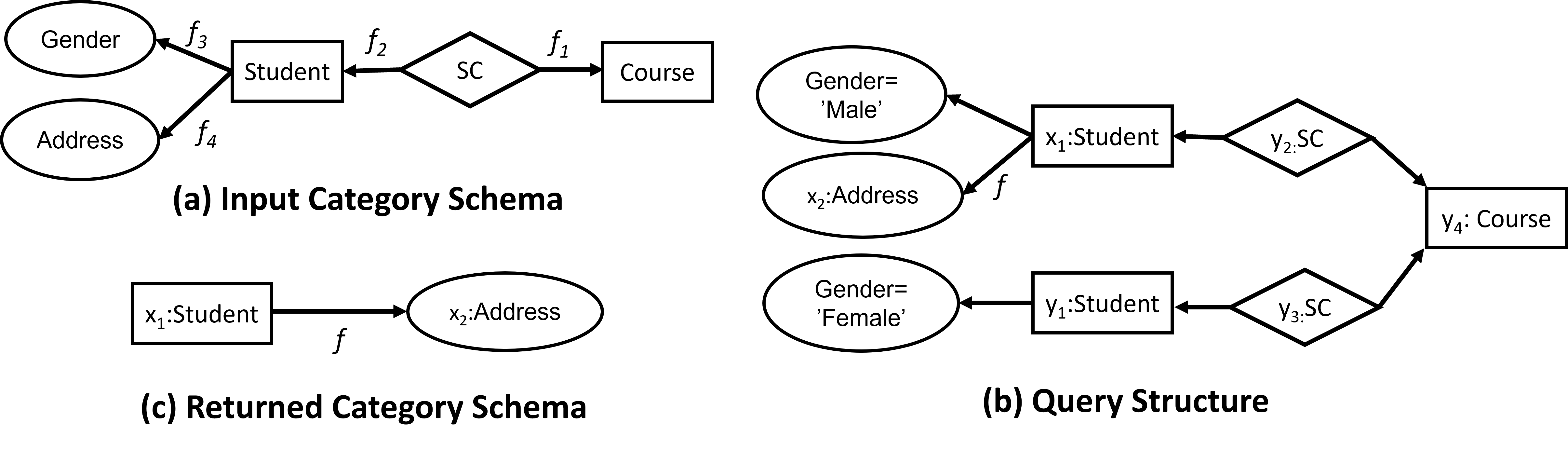}\caption{An illustrative example for student course category }\label{fig:SC2}\end{figure}













 \end{example}
 






 \begin{figure}\centering\includegraphics[width=0.7\textwidth]{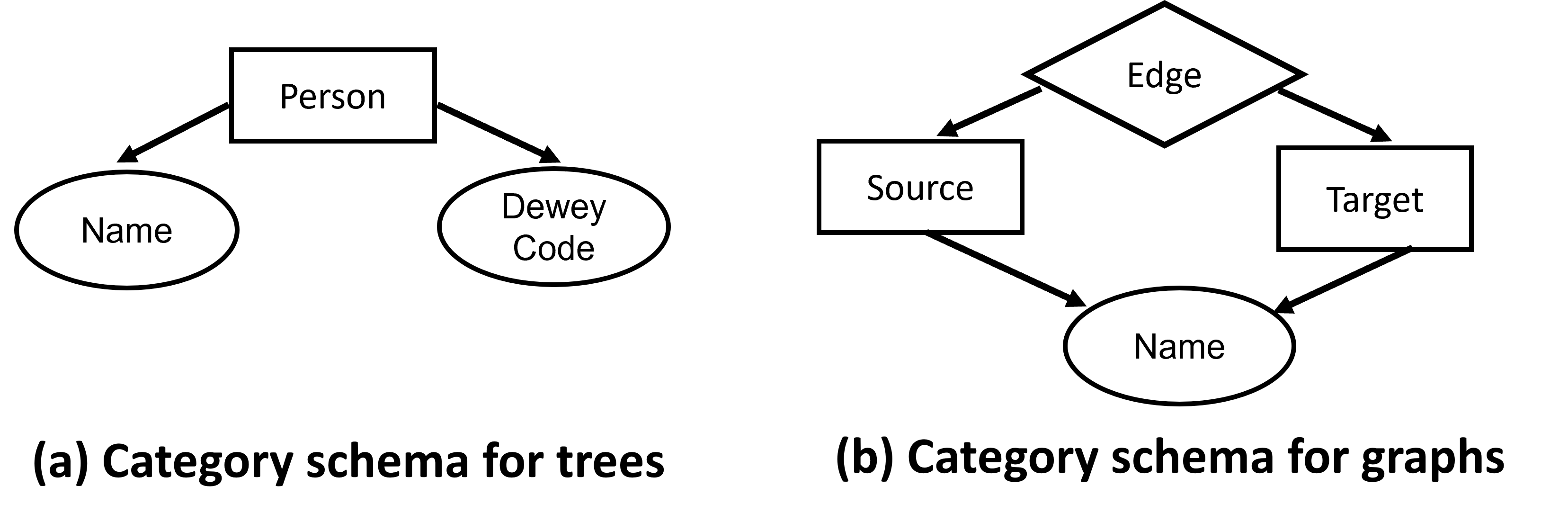}\caption{Categorical schemata for trees and graphs}\label{fig:friendexample2}\end{figure}

\begin{example} Imagine a family tree structure. Figure \ref{fig:friendexample2}(a) depicts a category schema of individuals, where Dewey codes are used to encode their positions in the tree. Consider a specific query: \textit{finding the names of all ancestors of ``John”}. The categorical calculus for this query is formulated as follows:


\{  $ x | x \in$ Name$\wedge \exists y_1$$\in$Person, $\exists y_2$$\in$Person, ($y_1 \cdot$DeweyCode $isAncestor$ $y_2\cdot$DeweyCode) $\wedge$  ($y_1 \cdot$Name = $x$)  $\wedge$ ($y_2 \cdot$Name=``\texttt{John}”)\}

\end{example}

 \begin{example}  Figure \ref{fig:friendexample2}(b)  illustrates a category schema for graph-structured data. The ``\texttt{Edge}” object represents a relationship object that projects onto the source and target of each edge. This example is used to demonstrate the reachable predicate. Consider the following query: \textit{find the names of all persons reachable by ``John”.}

\{ $x$ $|$ $x \in$ Name $\wedge$ $\exists y_1$$\in$ Source, $\exists$$y_2$$\in$Target, ($y_1  \rightsquigarrow^{Edge} y_2$) $\wedge$  ($y_2 \cdot$Name=$x$)  $\wedge$   ($y_1 \cdot$Name=``\texttt{John}”)\}

\end{example}

\section{Categorical Algebra}
\label{sec:algebra}



Categorical algebra is a formal system for manipulating and querying objects and morphisms within categories. It provides a set of operations that take one or more objects and morphisms as input and produce a new object or category as output. These operations are divided into two classes: \textit{set operations} and \textit{category operations}. Set operations perform actions on a single set or multiple sets (objects), while category operations involve transforming objects and morphisms into a category (categorization) and transforming a category into an object (with limit). These operations enable the precise and structured querying and manipulation of data within categories.


\subsection{Set Operators}

A set operator can be a unary, binary, or ternary operator.

\subsubsection{Unary Set Operators}  Three operators: Map, Project, and Select are introduced as follows.


\noindent 1. Map ($f$ or $\cdot$)

The map operator transforms one set into another using a function. As discussed in Section 3, there are two ways to present functions (morphisms) in this paper. Given an object $S_0$,  and functions $f_i: S_{i-1} \to S_i$, $1 \leq i \leq n$, the map operator can be written in either of the following forms:

\[f_n(...(f_1(S_0))) \myeq \{ y  |  x \in S_0 \wedge y = f_n(...(f_1(x))) \}, ~or\]
\vspace{-1mm}
\[ S_0 \cdot S_1... \cdot S_n \myeq \{ y  |  x \in S_0 \wedge y = x \cdot S_1... \cdot S_n \}\]

\noindent 2. Project ($\pi$)

The projection operator in categorical algebra  resembles the one in relational algebra applied to relational objects. Let  $A=(A_1,A_2,...,A_k)$ be a list of  projected components in a relationship object $R$.  Given any $r \in R$, for $i = 1,2,...,n (n \geq k)$ the notation $r[A_i]$ denote the $i$-$th$ component $A_i$ of $r$.

\[r[A] ~ \myeq  ~(r[A_1],r[A_2],...,r[A_k])\]

The projection of $R$ on $A$ is defined by

\[\pi_A R ~ \myeq ~ \{ r | r[A] \in R \}\]



\noindent 3.  Select ($\sigma$) 

The select operator is a unary operator which filters elements based on the given condition.  Selection expression has the form of $\sigma_{A \theta_M B} (S)$, where $A$ and $B$ represent two function expressions, and  $\theta_M$ is a mathematical predicate including $=, \neq, <, >, \le, \ge$. \[\sigma_{A \theta_M B}(S)  \myeq  \{ x  |  x \in S \wedge (A(x)  \theta_M B(x)  \}\], where $A(x)$  is the evaluation of a function expression $A$ over $x \in S$ such that \[
    A(x)  = 
\begin{cases}
    f_n (\cdots f_1(x)), & \text{if} ~ A  = f_n \circ \cdots \circ f_1 \\ 
    x \cdot S_1 \cdot ... \cdot S_n.   & \text{if} ~ A  = S_1 \cdot ... \cdot S_n
\end{cases}
\]
and $B(x)$ is defined in a similar manner. 





\begin{example} Recall the category schema depicted in Figure \ref{fig:SC2}(a). Consider the following query: \textit{select the courses taken by at least one female student and one male student}.

$S_1 = \sigma_{student.gender='Female'}(SC)$ (Select operator)

$S_2 = \sigma_{student.gender='Male'}(SC)$ (Select operator)

$S_3 = S_1 \cap S_2 $ (Binary operator)

$S_4 = S_3 \cdot Course$ (Map operator)

Combine them together: $ (\sigma_{student.gender='Female'}(SC) \cap \sigma_{student.gender='Male'}(SC)) \cdot Course $

 \end{example}

\subsubsection{Binary Set Operators} \label{sec:divisiontree}

The binary set operators union ($\cup$), intersection ($\cap$), difference ($-$) and Cartesian product ($\times$) are defined in the usual way between sets.  In the following, we introduce two additional binary set operators about division and tree operators.



1. Division ($\div$)

The division operator in categorical algebra is similar to that in relational algebra \cite{DBLP:persons/Codd72}. In particular, consider dividing a relationship object $R$ of degree $m$ by an object $S$ of degree $n$ with respect to the components of $A$ on $R$ and $B$ on $S$ respectively, denoted by $R[A] \div S[B]$. The degree of the division object is $m-n$. \[ R[A] \div S[B] \myeq \{ (x_1,\ldots,x_n) | (x_1, \ldots, x_n)  \in \pi_{\overline{A}} R  \wedge \forall (y_1, \ldots, y_m) \in  \pi_{B} S´,  (x_1, \ldots, x_n, y_1, \ldots, y_m) \in R  \},\] where $\overline{A}$ and $\overline{B}$ are the component list of $R$ and $S$ that is complementary to $A$ and $B$ respectively. If $n=1$, then $B$ can be omitted in the division expression. 

The division operator can be implemented with other operators:

$ R[A] \div S[B] =  \pi_{\overline{A}} R -\pi_{\overline{A}} ( (\pi_{\overline{A}} R \times \pi_B S) - R  )  $

 \begin{example}  Refer again to Figure \ref{fig:SC2}(a). Consider the query: \textit{find the addresses of all male students who attend all courses attended by any female student}.

$S_1 = \sigma_{student \cdot gender='Female'}(SC)$

$S_2= \sigma_{student \cdot gender='Male'}(SC)$ 

$S_3 = S_2[course] \div (S_1 \cdot course) $

$S_4 = S_3 \cdot address$

Combine them together:

(($\sigma_{student \cdot gender='Male'}SC)[course]$ $\div$ ($\sigma_{student \cdot gender='Female'}SC) \cdot course$) $\cdot$ address 
 \end{example}

2. Binary operators with tree data

Consider a category $\mathcal{T}$ that exhibits a tree structure. Given two sets of Dewey codes $D_1$ and $D_2$ in \( \mathcal{T} \), the function \( getParent(D_1, D_2) \) identifies all pairs of elements that have parent-child relationships.

\vspace{-3mm}
\[getParent(D_1,D_2) \myeq \{ (x,y) | x \in D_1 \wedge  y\in D_2  \wedge (x ~ is ~ a ~ \text{prefix} ~ of ~ y)  \wedge \text{level}(x)=\text{level}(y)-1\}\]

Similarly, we can define the operator $getAncestor$ to find the pairs with ancestor-descendant relationships. \[getAncestor(D_1,D_2) \myeq \{ (x,y) | x \in D_1 \wedge  y\in D_2  \wedge (x ~   \text{is a prefix of}  ~ y)  \}\]



 


 

Additional tree operators, such as $getSibling$, $getPreceding$ and $getFollowing$, among others,  can be found in Appendix \ref{sec:otheralgebra}.

\subsubsection{Ternary operators with graph data}  


 Consider a category $\mathcal{G}$ with graph structure described in Figure \ref{fig:friendexample2}(b). Given two node sets $Source$ and $Target$, and one edge set $Edge$ in $\mathcal{G}$. The operator $getReach$($S,T,E$) returns pairs of elements from $S$ and $T$ which are reachable via the edges in $E$.  
\[ getReach(S,T,E) \myeq \{ (x_1, x_2) | x_1 \in S  \wedge x_2 \in T  \wedge x_1 \rightsquigarrow^E x_2 \} \]






Similarly, we can define the $n$-$Hop$ operator as follows: \[ getnHop(S,T,E) \myeq \{ (x_1, x_2) | x_1 \in S  \wedge x_2 \in T  \wedge x_1 \rightsquigarrow^E_n x_2 \} \]


\begin{example} Consider a query to find all recursive friends of ``\texttt{John}” in Figure \ref{fig:friendexample2} (b).
 
$S_1 = \sigma_{Name=``John"} Source$

$S_2 = getReach(S_1,Target,Edge)$
 
$S_3 = (\pi_{Target} S_2) \cdot Name$

Combined together:

$(\pi_{Target} (getReach(\sigma_{Name=``John"} Source,Target,Edge))) \cdot Name$

 \end{example}


The computation of $getReach$ is similar to calculating the \textbf{transitive closure}, a well-studied research problem (e.g. \cite{schnorr1978algorithm,ioannidis1988efficient,agrawal1990direct}). The key difference is that while the transitive closure computes all reachable pairs for every node, $getReach$ specifically returns the pairs between the sets \( S \) and \( T \). The algorithms used to compute $getReach$ are detailed in Appendix \ref{sec:getreachalgorithm}.




\subsection{Category operators}

Category operators perform the transformation between sets, functions and categories.

\subsubsection{Sets and Functions to Category}

Given a set of sets $S_1$, $\cdots$, $S_n$ and functions defined on those sets, $f_1:S_{i_1} \to S_{j_1}$, $\cdots$ , $f_m:S_{i_m} \to S_{j_m}$, we have  \[Cat(S_1,...,S_n,f_1:S_{i_1} \to S_{j_1},..., f_m:S_{i_m} \to S_{j_m})\]


This operator, called \textbf{Categorification}, constructs a category using a given set of objects and morphisms. Note that this construction is constrained to be applied only to  well-defined functions  $f: S_i \to S_j$, such that for each element $ x \in S_i $, there must exist exactly one corresponding image $f(x)$ in the set $S_j$.

\subsubsection{Category to Set} \label{sec:sublimit}

We introduce a new operator called \textbf{Limit} which converts a category into a relational object (set).  

\[Lim(Cat(S_1,...,S_n,f_1:S_{i_1} \to S_{j_1},..., f_m:S_{i_m} \to S_{j_m})) \] 
\[\myeq  \{(x_1, \cdots , x_n)|  x_i \in S_i, 1 \leq i \leq n \wedge   f_k: x_{i_k} \to x_{j_k}, f(x_{i_k})=x_{j_k}, 1 \leq k \leq m \}\]

The limit operator, defined on a set of objects \( \mathcal{S} = \{x_1, \cdots, x_n\}\) and morphisms \( \mathcal{F} = \{ f_1, \cdots, f_m\} \), generates a relationship object containing elements from each object in \( \mathcal{S} \) such that those elements satisfy the respective function mappings in \( \mathcal{F} \). In category theory, limits play a crucial role in analyzing and defining structures within categories.  The purpose of the limit operator resembles that of the join operator in relational databases: it joins multiple objects (relations) into a single object (relation) to satisfy the corresponding functional mappings (join conditions).

We show the equivalence between categorical algebra and categorical calculus as follows:


\begin{theorem}
The categorical calculus  and categorical algebra are equivalent.    
\label{the:equivalent}\end{theorem}

\smallskip

Proof: The proof can be accomplished using the following two steps:

(i) Algebra queries $\sqsubseteq$ calculus queries; and

(ii) Calculus queries $\sqsubseteq$ algebra queries.

For (i) it is sufficient to show that each of the algebraic operations can be simulated by a calculus expression. The simulation of algebra operations by calculus expression is as follows:

1. Map $f(S_1)$, $f: S_1 \to S_2$ : \[ \{y | y \in S_2 \wedge \exists x \in S_1, y = f(x) \}; \]

2.  Select  $\sigma_{A \theta_M B}(S)$: \[ \{x |x \in S \wedge A(x) \theta_M B(x) \} \]

3. Project $\pi_A R$, where $A= A_1,...,A_k$ is a sublist of components of a relationship object $R$.  \[ \pi_A R = \{ (x_1,...,x_k) |  x_i \in A_i (1 \leq i \leq k) \wedge \exists (x_1,...,x_k,x_{k+1},..., x_n) \in  R \}\]

4. $S_1[A] \div S_2[B]$: \[ \{ (x_1,\ldots,x_n) |  (x_1, \ldots, x_n) \in \pi_{\overline{A}} S_1  \wedge (\forall (y_1, \ldots, y_m) \in  \pi_{B} S_2´,  (x_1, \ldots, x_n, y_1, \ldots, y_m) \in S_1)  \}\]

5. $getParent$($D_1$, $D_2$), where  $D_1$ and $D_2$ are two Dewey code sets.  \[ \{(x_1,x_2) | x_1 \in D_1 \wedge x_2 \in D_2 \wedge x_1 ~ isParent ~ x_2\} \]


6. $getReach(S,T,E)$, where $S$ and $T$ denote the source and target node set, and $E$ the edge set. \[ \{(x_1,x_2) | x_1 \in ~ S \wedge x_2 \in ~ T \wedge x_1  \rightsquigarrow^E x_2  \}; \]

7. $getnHop$($S,T,E$):  \[ \{(x_1,x_2) | x_1 \in ~ S \wedge x_2 \in ~ T \wedge x_1 \rightsquigarrow_n^E x_2  \};
\]





8. $Cat(S_1,\cdots,S_n, f_1: S_{i_1} \to S_{j_1},\cdots,f_m: S_{i_m} \to S_{j_m})$: \[ \{x_1,\ldots,x_n | x_i \in S_i (1 \leq i \leq n) \wedge f_k:S_{i_k} \to S_{j_k}(1 \leq k \leq m) \} \]

9. $Lim(Cat(S_1,\cdots,S_n, f_1:S_{i_1} \to S_{j_1},\cdots, f_m: S_{i_m} \to S_{j_m}))$: \[ \{(x_1, \cdots , x_n) |  x_i \in S_i (1 \leq i \leq n) \wedge f_{k}(x_{i_k}) = x_{j_k} (1 \leq k \leq m)  \} \]


\smallskip

Then we proceed to prove that (ii) calculus queries $\sqsubseteq$ algebra queries.


We will exhibit an algorithm for translating a calculus expression into a semantically equivalent algebraic expression. The reduction algorithm can be outlined as follows: First, the categorical calculus expression is converted into prenex normal form where the quantifier-free part (matrix) is disjunctive normal form. Second, for each conjunctive clause, a category is constructed to include all variables and their associated morphisms. Third, the limit of each category is computed, and elements that do not satisfy the commutative diagram of the limit are naturally removed. Fourth, the remaining elements in the limit are further refined through selection and division. Finally, the target objects and morphisms are projected to construct the desired result. A  detailed account follows.

Step 1: 

(1.1) Convert $V$ to the prenex normal form without expanding the range-coupled qualifiers;

(1.2) Convert the matrix (the part of the formula that remains after all quantifiers) in $V$ into disjunctive normal form;

(1.3)  Apply a change to variables resulting from (1.2) so that variables become $x_1$,...,$x_{n}$ in the order of their first occurrence in the qualification. 


Let a calculus expression resulting from the above transformations be 

\vspace{-4mm}
\[ Z' = \{  X', \mathcal{R'} | U' \wedge F'  \wedge V' \} \]


Step 2: Generate the following three types of objects and morphisms;

(2.1) Generate entity or attribute objects for each variable in $X'$ and $\mathcal{R'}$, For each variable $x_i$, form a defining equation of the range term of  $x_i$. The algebraic expression on the right-hand side of this equation is obtained by applying the following rewriting rules to $x_i$:

(i) $ x_i \in O \Rightarrow S_i = O$
(ii) $\vee \Rightarrow \cup$ (set union)
(iii) $\wedge \Rightarrow \cap$ (set intersection)
(iv) $\wedge \neg \Rightarrow -$

For example, if $x_i \in O_1 \wedge  \neg x_i \in O_2$, then the defining equation for $S_i$, 
\vspace{-1mm}
\[ S_i = O_1-O_2\]

In this case, $O_1$ and $O_2$ must be union-compatible. 


(2.2) Generate relationship objects for each set $R_i$ of variables in $\mathcal{R'}$. Specifically, for each \( R_i = (x_{i_1}, \ldots, x_{i_m}) \), create objects \( R_i \) and define the associated projected morphisms \( \pi_j: R_i \to S_{i_j} \) for each \( j \) (\( 1 \leq j \leq m \)) where $x_{i_j} \in S_{i_j}$.  

(2.3) Generate relationship objects for all tree and graph predicates in $V'$, including reachable ($\rightsquigarrow^E$), nHop ($\rightsquigarrow_n^E$), $isParent$, and other tree operators.

(i)  $x_1 \in S, x_2 \in T, x_1 \rightsquigarrow^E x_2$ $\Rightarrow$  Create $R_i = getReach(S,T,E)$, $\pi_1: R_i \to S$ and $\pi_2: R_i \to T$;

(ii) $x_1 \in S, x_2 \in T, x_1 \rightsquigarrow^E_n x_2$ $\Rightarrow$  Create $R_i = getnHop(S,T,E)$, $\pi_1: R_i \to S$ and $\pi_2: R_i \to T$;

(iii) $x_1 \in D_1, x_2 \in D_2, isParent(D_1,D_2)$ $\Rightarrow$ Create $R_i = getParent(D_1,D_2)$, $\pi_1: R_i \to D_1$, $\pi_2: R_i \to D_2$.

Step 3: 

(3.1) For each clause $C_k$ in the disjunctive normal form of $V'$, generate morphisms such that ($f(x_i)=x_j$ or $x_i \cdot S_j = x_j$) $\Rightarrow f: S_i \to S_j$, where $x_i$ and $x_j$ appear in the clause $C_k$.

(3.2) Construct one category $ \mathcal{C}_k$ for each clause, involving the associated objects and morphisms from Steps (2) and (3.1), $ \mathcal{C}_k= Cat(S_1,\cdots,S_n, f_1:S_{i_1} \to S_{j_1},\cdots, f_m: S_{i_m} \to S_{j_m})$;

(3.3) Compute a universal limit object $L_k = Lim(Cat(\mathcal{C}_k)$, where   $\mathcal{C}_k$ is the category in Step (3.2).

Step 4: For each clause $C_k$ in the disjunctive normal form of $V'$, continue to apply the select operations for each limit object $L_k$ generated from step (3).

(i) $A(x_i) \theta_M  B(x_j) \Rightarrow \sigma_{A(x_i) \theta_M  B(x_j) }(L_k)$
(ii) $A(x_i) \theta_M  \alpha \Rightarrow \sigma_{A(x_i) \theta_M  \alpha }(L_k)$

where $a$ is an individual constant, $\theta_M$ is one of $=,\neq,>,<,\leq,\geq$. 



Step 5: Compute the division on the limit object obtained in the previous step to incorporate \( S_j \) ($ 1 \leq j \leq n $) corresponding to the universal quantifier \( \forall \) of \( V' \). The division operation is defined as \( L[S_1,\cdots,S_n] \div (S_1 \times \cdots \times S_n)[S_1,\cdots,S_n] \), where \( L \) represents the limit object from Step (4), and \( S_1 \times \cdots \times S_n \) denotes the Cartesian product of sets. After performing division for each clause, combine the results from individual clauses using a union operation.

Step 6:  Form projected objects that take into account the target objects in $X'$ and $\mathcal{R'}$. Finally, construct the desired target category.



The validity of the reduction algorithms relies on the following lemmas. 


\begin{lemma}\label{lem:R} Consider a category $\mathcal{C}$ with a set of objects $S_1, S_2, \cdots, S_n$ and  functions $f_1:S_1 \to S_2, \cdots, f_n:S_n \to S_{n+1}$. Let a relational object $R=Lim(\mathcal{C})$. We have $\forall x_i \in \pi_{S_i} (R)$, $f(x_i) \in S_{i+1}$, where $1 \leq i \leq n$.
\end{lemma}

\begin{proof}
Based on the definition of limit operator in Section \ref{sec:sublimit}: 
\[R =  \{(x_1, x_2, \cdots, x_{n+1} )|  x_1 \in S_1, \cdots, x_{n+1} \in S_{n+1}   \wedge   f(x_{1})=x_{2}, \cdots, f(x_{n})=x_{n+1} \}\]

Therefore,  $\pi_{S_i} (R) = S_i$ and $\forall x_i \in S_i $, $f(x_i) =x_{n+1} \in S_{n+1}$, where $1 \leq i \leq n$.

\end{proof}


\begin{lemma}\label{lem:division} This lemma addresses the validity of the computation of division for multiple  $\forall$ quantifiers. Consider the following calculus expression: \begin{equation}
\begin{split}
\{ x_1 | x_1 \in S_1  \wedge \forall y_1 \in S_2,  \cdots, \forall y_n \in S_{n+1}, \exists z_1 \in R_1, \cdots, \exists z_n \in R_n, f_1(z_1)=x_1 \wedge  g_1(z_1)=y_1,  \\
              \cdots \wedge f_n(z_n)=x_n \wedge g_n(z_n)=y_n \}
\end{split}
\end{equation}

This expression is equivalent to the following algebraic operators: \[\pi_{S_1}(L[S_2,S_3, \cdots, S_{n+1}] \div (S_2 \times S_3 \times \cdots S_{n+1} )[S_2,S_3 \cdots, S_{n+1}])\]  where $L=Lim(Cat(S_1,\cdots,S_{n+1},R_1, \cdots, R_n,f_1:S_1 \to S_2, \cdots, f_n:S_1 \to S_{n+1}))$.

\end{lemma}

\begin{proof}
Based on the definition of division operator in Section \ref{sec:divisiontree},

$\pi_{S_1}(L[S_2,S_3, \cdots, S_{n+1} ] \div (S_2 \times S_3 \times \cdots S_{n+1} )[S_2,S_3, \cdots, S_{n+1}])$ returns \[ \{ x_1 | x_1 \in S_1 \wedge  \exists z_1 \in R_1, \cdots, \exists z_n \in R_n, \forall (y_1, y_2, \cdots, y_{n}) \in  \pi_{S_2,S_3, \cdots, S_{n+1}}(R),  (x_1, y_1, \cdots, y_{n}, z_1, \cdots, z_n ) \in R  \}\]

Note that $L$ is a limit object computed through the category $\mathcal{C}$=$Cat$($S_1,\cdots,S_{n+1},R_1, \cdots, R_n,f_1:S_1 \to S_2, \cdots, f_n:S_1 \to S_{n+1}$). Thus, based on Lemma \ref{lem:R},  the above expression can be rewritten as 
$\{ x_1 | x_1 \in S_1  \wedge \forall y_1 \in S_2, \cdots, \forall y_n \in S_{n+1}, \exists z_1 \in R_1, \cdots, \exists z_n \in R_n,  f_1(z_1)=x_1 \wedge  g_1(z_1)=y_1,$ $\cdots \wedge f_n(z_n)=x_n \wedge g_n(z_n)=y_n \}$, as desired.

\end{proof}

\begin{lemma}\label{lem:graphoperator} This lemma addresses the validity of the reachability operator for graph data. Consider the following calculus expression: $\{ x_1, x_2 | x_1 \in S_1  \wedge x_2 \in S_2  \wedge x_1 \rightsquigarrow^E x_2 \wedge \exists x_3 \in S_3, f_1(x_3)=x_1 \wedge f_2(x_3)=x_2\}$. This calculus can be implemented with the following three algebraic operators: $S_4$= $getReach(S_1,S_2,E)$; $S_5$ =  $Lim(Cat(S_1,S_2,S_3,S_4,f_1: S_3 \to S_1,f_2: S_3 \to S_2,\pi_1: S_4 \to S_1,\pi_2: S_4 \to S_2))$ and $S_6$ = $\pi_{S_1,S_2}S_5$. 
\end{lemma}

\begin{proof}
The $getReach$ operator returns the pairs of elements from $S$ and $T$ which are reachable via the edges in $E$. 

$S_4$= $\{ x_1, x_2 | x_1 \in S  \wedge x_2 \in T  \wedge x_1 \rightsquigarrow^E x_2 \}$,

Based on the definition of limit operator in Section \ref{sec:sublimit}: 
\[S_5 =  \{(x_1, x_2,x_3, x_4 )|  x_1 \in S_1, x_2 \in S_2, x_3 \in S_3, x_4 \in S_4,   f_1(x_{3})=x_{1} \wedge f_2(x_{3})=x_{2} \wedge \pi_{S_1}(x_{4})=x_{1} \wedge \pi_{S_2}(x_{4})=x_{2} \}\]

Apply the definition of the reachability predicate: \[S_5 =  \{(x_1, x_2,x_3, x_4 )|  x_1 \in S_1, x_2 \in S_2, x_3 \in S_3, x_4 \in S_4,   f_1(x_{3})=x_{1} \wedge f_2(x_{3})=x_{2} \wedge (x_1 \rightsquigarrow^E x_2) \}\]

Finally, perform the projection operator on $x_1$ and $x_2$: \[S_6 =  \{(x_1, x_2 )|  x_1 \in S_1, x_2 \in S_2,  \exists x_3 \in S_3, f_1(x_{3})=x_{1} \wedge f_2(x_{3})=x_{2} \wedge (x_1 \rightsquigarrow^E x_2) \}\]

as desired.

\end{proof}

 \begin{example} Recall the calculus expression in Example \ref{exp:calculus}. We use this example to illustrate the translation algorithm above.



Step 1: Convert it into the prenex normal form and convert the matrix into disjunctive normal form.

\{ $x_1$, $x_2$ $|$ $x_1$ $\in student \wedge x_2 \in address \wedge  $ ($f: x_1$ $\to$ $x_2  = f_2$) $\wedge$  $ \forall y_1 \in student,\exists y_2 \in SC, \exists y_3 \in SC, \exists y_4 \in course,
( (x_1 \cdot gender=\text{``Male”}  \wedge y_1 \cdot gender \neq \text{``Female”}) \vee   ( (x_1 \cdot gender=\text{``Male”})  \wedge (y_2 \cdot student= x_1) \wedge (y_3 \cdot student= y_1) \wedge  (y_2 \cdot course = y_4) \wedge (y_3 \cdot course = y_4) ) )  $  \}

Step 2: Generate the objects for each variable $x_1$, $x_2$, $y_1$ to $y_4$:

$S_1 = student$ for $x_1$; $S_2 = address$ for $x_2$; $S_3 = student$ for $y_1$; $S_4 = SC$ for $y_2$; $S_5 = SC$ for $y_3$; and $S_6 = course$ for $y_4$.

Step 3: Construct a category for each clause in disjunctive normal form, involving the relevant objects and morphisms, and then compute the corresponding limit:

The first limit: \( S_7 = \text{Lim}(S_1, S_3) \), which represents the Cartesian product of \( S_1 \) and \( S_3 \).

The second limit: \( S_8 = \text{Lim}(S_1, S_3, S_4, S_5, S_6, f_1: S_4 \to S_1, f_2: S_5 \to S_3, f_3: S_4 \to S_6, f_4: S_5 \to S_6 ) \).

Step 4: Apply the select condition for the limit objects:

$S_9=  \sigma_{S_1.gender='Male'}(S_7) \cap \sigma_{S_3.gender \neq 'Female'}(S_7)$

$S_{10}=  \sigma_{S_1.gender='Male'}(S_8)$

Step 5:   Compute the division for each clause with respect to the variable \(\forall y_1\) (i.e. $S_3$), then union the results.

$S_{11} =   S_9[student] \div S_3 $, which returns empty; 

$S_{12} =   S_{10}[student] \div S_3 $;

$S_{13} = S_{11} \cup S_{12}$ = $S_{12}$.

Step 6: Return the final result 

$S_{14} = \pi_{S_1}S_{13}$; $S_{15} = \pi_{S_2}S_{13}$; 
 $S_{16} = Cat(S_{14}, S_{15}, f_3 : S_{14} \to S_{15})$. Return $S_{16}$.

 \end{example}

\smallskip

\section{Algebraic Transformation Rules }
\label{sec:rules}





Query optimization in relational databases utilizes various algebraic transformation rules. These rules are applied to transform the original query into an equivalent query that can be executed more efficiently. Similarly, in the context of multi-model databases, the rules of categorical algebraic transformation can be used for multi-model query optimization.  We present a family of rewriting rules that can be used by multi-model query optimizers with algebraic transformation.




(1)  \textbf{Cascade of $f$}: A sequence of functions can be concatenated as a composition of individual function operations.

\[ {f_n}   ( \ldots ({f_1}(S))   \equiv {(f_n \circ \ldots \circ f_1)}(S)    \]



 (2) \textbf{Lim and $\pi$ }: The project operator ($\pi$) functions as the inverse of the $limit$ object operator. It retrieves specific components from a relationship object, contrasting with the limit operator which synthesizes elements from multiple objects.

  \[ \pi_{S_1}(Lim(Cat(S_1,S_2,f:S_1 \to S_2))) \equiv S_1\]

  \[ \pi_{S_2}(Lim(Cat(S_1,S_2,f:S_1 \to S_2))) \subseteq S_2\]

In particular,  if there is no morphism between $S_1$ and $S_2$, then $Lim(S_1,S_2)$=$S_1 \times S_2$. In this case, both equations hold.

  \[ \pi_{S_1}(Lim(S_1,S_2)) \equiv S_1 ~ \text{and} ~ \pi_{S_2}(Lim(S_1,S_2)) \equiv S_2\]

 (3) \textbf{Pushing $\sigma$ to one or multiple objects in $lim$}:

  \[ \sigma_C(Lim(S_1,S_2)) \equiv Lim(\sigma_C(S_1),S_2)\]

  Alternatively, if the selection condition $C$ can be written as $C_1 \wedge C_2$, where condition $C_1$ involves only the values of $S_1$ and condition $C_2$ involves only the values of $S_2$, the operation commutes as follows:

\[ \sigma_C(Lim(S_1,S_2)) \equiv Lim(\sigma_{C_1}(S_1),\sigma_{C_2}(S_2)\]

(4)    \textbf{Pushing $\sigma$ to one or multiple objects in }$getReach$:

 Assume that $C_1$ is a filtering selection on the set $S$, and $C_2$ is on $T$, then the following \textit{pushdown} operator is valid: \[ getReach(\sigma_{C_1}(S),\sigma_{C_2}(T),E) \equiv \sigma_{C_1 \wedge C_2}(getReach(S,T,E)) \]

  (5) \textbf{Pushing $\sigma$ to one or two objects in $getParent$ and other tree-structure operators.} \[ \sigma_C(getParent(D_1,D_2)) \equiv  getParent(\sigma_{C_1}(D_1), \sigma_{C_2}(D_2)) \]



(6) \textbf{Commuting function mapping with the product operator}.





Consider a category \(\mathcal{C}\). If \(f: A \to B\) and \(g: C \to D\) are two morphisms in \(\mathcal{C}\), then their  product \(f \otimes g: A \times C \to B \times D\) is also a morphism in \(\mathcal{C}\). Given two sets \( S_1 \) and \( S_2 \), we have \[ (f \otimes g)(S_1 \times S_2) \equiv f(S_1) \times g(S_2) \]


(7) \textbf{Commuting $\pi$ with the $Lim$ operation}. Given two relationship objects $R_1$ and $R_2$ and a function $f$: $R_1 \to R_2$, and a projection List $L$=$\{A_1,...,A_n,B_1,...,B_m\}$, where $L_1=\{A_1,...,A_n\}$ are components of $R_1$ and  $L_2=\{B_1,...,B_m\}$ are components of $R_2$. If there exists another function $f_2: R'_1 \to R'_2 $ where $R_1'=\pi_{A_1,...,A_n}(R_1)$ and $R_2'=\pi_{B_1,...,B_m}(R_2)$.
 such that the following diagram commutes:



\[
\begin{tikzcd}
R_1 \arrow[r, "f_1"] \arrow[d, "\pi_{L_1}"'] & R_2 \arrow[d, "\pi_{L_2}"] \\
R_1' \arrow[r, "f_2"'] & R_2'
\end{tikzcd}
\]

 then the two operators $\pi$ and $Lim$ can be commuted as follows:\[ \pi_L(Lim(Cat(R_1, R_2, f_1:R_1 \to R_2))) \equiv Lim(Cat(\pi_{L_1}(R_1), \pi_{L_2}(R_2), f_2: \pi_{L_1}(R_1) \to \pi_{L_2}(R_2)) ) \]


(8) \textbf{Commuting $g$ with the $Lim$ operation}.  Consider a category \(\mathcal{C}\) with two objects \(S_1\) and \(S_2\), and a morphism \(f_1: S_1 \to S_2\). Let \(g\) be a function with the domain \( \operatorname{Lim}(\mathcal{C}) \), which can be decomposed into two functions \(g = g_1 \otimes g_2\), where the domains of \(g_1\) and \(g_2\) are \(S_1\) and \(S_2\), respectively. If there exists another morphism \(f_2: g_1(S_1) \to g_2(S_2)\) such that the following diagram commutes:


\[
\begin{tikzcd}
S_1 \arrow[r, "f_1"] \arrow[d, "g_1"'] & S_2 \arrow[d, "g_2"] \\
S_1' \arrow[r, "f_2"'] & S_2'
\end{tikzcd}
\]

then the two operators $g$ and $limit$ can be commuted as follows:

\[ g(Lim(Cat(S_1, S_2, f_1:S_1 \to S_2))) \equiv Lim(Cat(g_1(S_1), g_2(S_2), f_2: g_1(S_1) \to g_2(S_2)) ) \]

(9) \textbf{Commuting $Lim$ with the $getReach$ operation}.  Consider the scenario where the contents of a single object are filtered using both the $limit$ and $getReach$ operators. The following rule demonstrates that the order of applying these filters can be swapped. \[ \pi_{S_1}(getReach(\pi_{S_1}\sigma_C(lim(S_1,S_2)) \equiv   \pi_{S_1}\sigma_C(Lim(\pi_{S_1} getReach({S_1},T,E),S_2)) \]

\begin{theorem} \label{the:express}
    
Categorical  calculus and categorical algebra can express all of the following:
    \begin{itemize}
        \item Relational calculus and algebra queries;
        \item Graph pattern matching and graph reachability queries;
        \item XML twig pattern queries. 
    \end{itemize}
    
\end{theorem}

\noindent \textbf{Proof}: (Sketch) Categorical algebra contains the main operations found in relational algebra, making it capable of expressing any relational query. The fundamental operations in relational algebra—\textit{Select}, \textit{Project}, \textit{Cartesian Product}, \textit{Union}, and \textit{Set Difference}—each have corresponding counterparts in categorical algebra. Note that Cartesian product is a specific instance of the \textit{limit} operator, in the case of the absence of morphisms between objects. Further, graph reachability can be addressed using reachable ($\rightsquigarrow^E$ and $\rightsquigarrow_n^E$) predicates, while graph pattern matching can be translated into \textit{Map}, \textit{Select}, and \textit{Limit} operators to find all matching nodes. Moreover, parent-child, ancestor-descendant, sibling and other tree relationships can also be implemented using the corresponding categorical predicates with the assistance of Dewey codes of XML elements.


\smallskip



\begin{theorem} \label{the:complexity}
Consider a category $\mathcal{C}$ with $p$ objects and $q$ morphisms. Let the maximum number of elements in any object of $\mathcal{C}$ be $n$.

\begin{enumerate}
    \item The data time complexity of computations for categorical calculus and categorical algebra in $\mathcal{C}$ is bounded by $O(q \cdot n^p)$.
    \item The space complexity of these computations is bounded by $\text{NSPACE}[\log n]$.
\end{enumerate}
\end{theorem}

\noindent \textbf{Proof}: Due to the equivalence between categorical calculus and categorical algebra, we analyze the time complexity of categorical algebra. The main operators are analyzed as follows:

(1) Map $f(S)$: the time complexity is $O(n)$, where $n$ denotes the number of elements in  $S$; 

(2) Project $\pi(R)$: the time complexity is $O(n log(n))$, where $n$ denotes the number of elements in $R$, as this operator requires the removal of the duplicate elements; 

(3) Select $\sigma_C(S)$:  the time complexity is $O(n)$, where $n$ denotes the number of elements in  $S$; 

(4) $getParent(D_1,D_2)$:  the time complexity is $O(|D_1| \cdot |D_2|)$;

(5) $getReach(S,T,E)$:  the time complexity is O($n^2$+$min(|S|, |T|) \cdot |E|$) time, where $n$ is the total number of nodes in the graph, as shown in Algorithm \ref{alg:OTC}.

(6) $Lim(\mathcal{C})$: assuming that the category $\mathcal{C}$ contains $p$ objects, $q$ morphisms and the maximum number of elements in one object is $n$. The total cost is $O(q \cdot n^p)$ in the worst case.

The most computationally expensive operator in the list above is \( Lim(\mathcal{C}) \). If an algebraic query involves \( k \) operators, the total computational cost is \( O(k \cdot q \cdot n^p) \). Given that \( k \) is typically small relative to other parameters in the dataset, the data complexity simplifies to \( O(q \cdot n^p) \).

Regarding space complexity, most operators, except for the $getReach$ operator, operate within LOGSPACE, similar to relational calculus. However, the $getReach$ query has a space complexity of \( O(\log^2 n) \), comparable to transitive closure \cite{10.1145/800141.804682}. Consequently, the overall space complexity of the operators falls under NSPACE[\(\log n\)].

\section{Conclusion and future work}



In this work, we propose two query mechanisms for categorical databases: categorical calculus and categorical algebra. Both mechanisms aim to extract subsets of elements within objects that meet specified query conditions. While category theory traditionally emphasizes abstract relationships and structures between objects, often overlooking their internal elements, our approach takes a different direction. We focus on how to extract subsets of elements within objects to tackle the challenges of multi-model data querying. This adaptation offers a novel perspective on applying category theory to the database field and addresses new research challenges that have not been thoroughly explored in existing literature.

In future work, we plan to extend categorical calculus and algebra to support additional query operations, such as shortest path queries and aggregation queries. Moreover, given the unifying nature and simplicity of the operators in categorical algebra, another avenue for future research is to develop holistic query optimization algorithms for multi-model data. 


\bibliographystyle{eptcs}
\bibliography{generic}

\appendix

\newpage





\section{Appendix: More operators for tree data in categorical algebra}
\label{sec:otheralgebra}

Given a Dewey code $d$, we use ``$|d|$” to denote the length of $d$, ``$d[i]$” to denote the $i$th integer of $d$ (starting from 1), and ``$d[i,j]$” to denote the substring from its $i$'th integer to its $j$'th integer. Given two Dewey code sets $D_1$ and $D_2$, we define the following algebraic operators:

\[getParent(D_1,D_2) \myeq \{ (x,y) | x \in D_1 \wedge  y\in D_2  \wedge x ~ is ~ a ~ \text{prefix} ~ of ~ y  \wedge |x|=|y|-1\}\]

\[getAncestor(D_1,D_2) \myeq \{ (x,y) | x \in D_1 \wedge  y\in D_2  \wedge x ~   \text{is a prefix of}  ~ y  \wedge x < y\}\]

\[getSibling(D_1,D_2) \myeq \{ (x,y) | x \in D_1 \wedge  y\in D_2  \wedge x[1,|x|-1] = y[1,|y|-1]  \wedge |x|=|y|\}\]

\[getFollowing(D_1,D_2) \myeq \{ (x,y) | x \in D_1 \wedge  y\in D_2  \wedge x > y\}\]

\[getPreceding(D_1,D_2) \myeq \{ (x,y) | x \in D_1 \wedge  y\in D_2  \wedge x < y \}\]

\[getPrecedingSilbing(D_1,D_2) \myeq \{ (x,y) | x \in D_1 \wedge  y\in D_2  \wedge x[1,n-1] = y[1,n-1]   \wedge x < y \}\]

\[getFollowingSilbing(D_1,D_2) \myeq \{ (x,y) | x \in D_1 \wedge  y\in D_2  \wedge x[1,n-1] = y[1,n-1]  \wedge x > y\}\]

\section{Algorithms for getReach }
\label{sec:getreachalgorithm}


We propose an algorithm to compute $getReach(S,T,E)$ operator in Algorithm \ref{alg:OTC}, where the final results are stored in a reachability matrix $R$, where \( R[i][j] = 1 \) if there is a path from vertex \( i \) to vertex \( j \), and \( R[i][j] = 0 \) otherwise. In Line 1, start with the adjacency matrix \( R \) as the initial adjacent matrix computed by $E$. For each node $n$ from the smaller cardinality set between $S$ and $T$, perform a DFS or  BFS search for $n$ to find all its reachable nodes to update the entries in $R$ (Lines 2 and 3). Finally, the results are returned based on $R$ with respect to the candidate pairs from $S$ and $T$. The time complexity of this algorithm is $O\bigl(n^2$+$min(|S|, |T|) \cdot |E|\bigr)$ time, where $n$ is the total number of nodes in the graph.


The $getReach(S, T, E)$ problem, which returns the pairs of nodes reachable from sets \( S \) and \( T \) via edges in \( E \), is a special case of \textit{transitive closure}, which has been extensively studied in the fields of algorithms (e.g.\cite{schnorr1978algorithm,10.1145/800141.804682}) and databases (e.g. \cite{ioannidis1988efficient,agrawal1990direct}). Transitive closure has a space complexity of \( NSPACE[\log n] \), indicating that its space requirements are manageable, it can be fully parallelized, and it supports efficient incremental evaluation. Although transitive closure is computationally challenging, comparable to matrix multiplication—with the best known bound being approximately \( O(n^{2.371552}) \) as of April 2024 \cite{williams2023newboundsmatrixmultiplication}—using matrix multiplication algorithms in practice is generally not worthwhile.





\begin{algorithm}
\caption{Computing $getReach$ operator}
\label{alg:OTC}
\KwIn{two node sets $S$ and $T$, and one edge set $E$} 
\KwOut{$getReach(S,T,E)$} 
\DontPrintSemicolon

Initialize a reachability matrix $R$ with $n \times n$;  \tcp*{$n$ is the number of nodes in $E$}

Let $M$ denote the set with the fewer number of nodes between $S$ and $T$;

Perform a DFS or BFS search for each node in $M$ to update the corresponding entries in the matrix $R$; 

\textbf{Return}  the reachability results based on $R$ with respective to the pairs from $S$ and $T$.

\end{algorithm}

\section{An example for multi-model data}


In this section, we demonstrate the proposed algebra and calculus using a multi-model data example. As shown in Figure \ref{fig:firstexample}, three types of data—relation, XML, and graph—are unified under a single category schema. Let's consider a query:\textit{ find the names of all customers who, through a chain of acquaintances, know "John," have purchased a toy product, and have a credit limit greater than 2000.} This query requires a comprehensive join across all three types of data.

Categorical calculus: 

\{  $x$ $| x \in CName, \exists  a \in$ Source,   $\exists b \in$ Target, $\exists l \in$ OrderLine, ($a \rightsquigarrow^{Knows} b$) $\wedge$ 
 $b$$\cdot$Customer$\cdot$CName = ``\texttt{John}''  $\wedge$ $a$$\cdot$Customer$\cdot$ Credit\_limit $>$ 2000  $\wedge$ $l$$\cdot$Order$\cdot$Customer=$a$$\cdot$Customer $ \wedge$ $l$$\cdot$Product$\cdot$PName=``\texttt{Toy}'' $ \wedge$  $x$ = $a$$\cdot$Customer $\cdot$CName\}

One approach to categorical algebra for this query is:

 $S_1 = getReach(Source,Target,Knows)$ 

 $S_2 = \sigma_{Target 
 \cdot Customer  \cdot CName='John' \wedge Source \cdot Customer \cdot Credit\_limit>2000}(S_1)$

 $S_3=\sigma_{Product \cdot PName='Toy'}(OrderLine)$

$S_4 = \pi_{Source} S_2$

 $S_5 = \sigma_{S_3 \cdot Order \cdot Customer=S_4 \cdot Customer} Lim( Cat( S_3, S_4))$
 
$S_6 = \pi_{Source}S_5 \cdot Customer \cdot CName$

 Return $S_6$

 Another way to write the algebraic operators is to swap the order between $limit$ and $getReach$ operators  (Rule 9) and push down the select operator for $getReach$ (Rule 4) for query optimization.

$S_1 = \sigma_{Customer \cdot Credit\_limit>2000}(Source)$

$S_2=\sigma_{Product \cdot PName='Toy'}(OrderLine)$

$S_3=\pi_{S_1} (\sigma_{S_1 \cdot Cusotmer = S_2 \cdot Order \cdot Customer} Lim(Cat(S_1,S_2)))$

 $S_4 = \sigma_{Customer \cdot CName='John'}(Target)$

$S_5 = getReach(S_3, S_4, Knows)$

$S_6 = (\pi_{Source} S_5) \cdot Customer \cdot CName $

Return $S_6$

\begin{figure}\centering\includegraphics[width=0.75\textwidth]{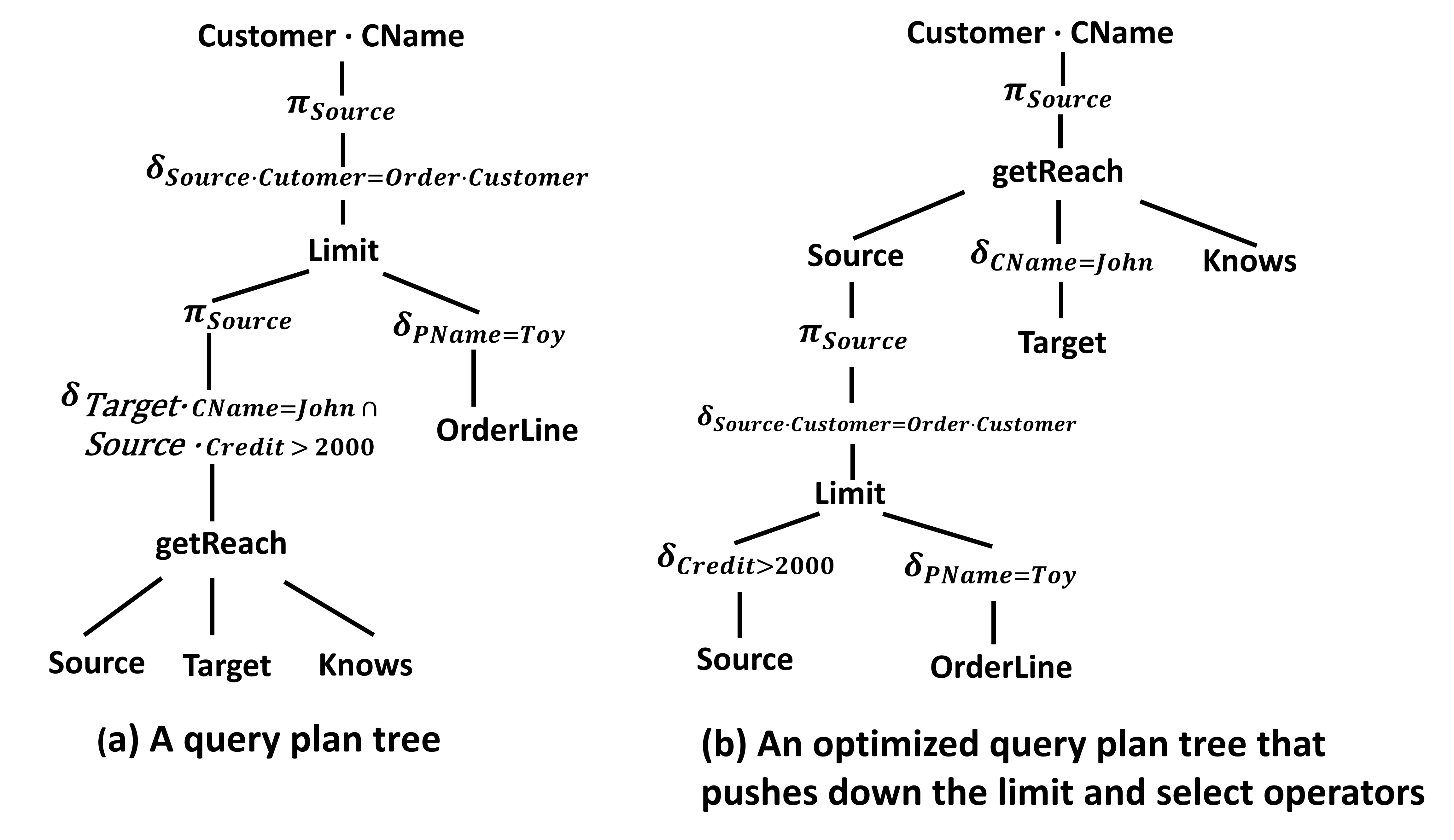}\caption{Two holistic query plans involving three types of data}\label{fig:queryplan}\end{figure}

 These two query plans are depicted in Figure \ref{fig:queryplan}. This example demonstrates that a single categorical calculus expression can be implemented using different sequences of categorical algebra operators. This framework paves the way for designing a holistic query plan for multi-model query optimization.  
 





\section{More algebraic transformation rules}
\label{sec:transformationmore}

In this section, we will first prove the transformation rules (7) and (8). We will then present additional transformation rules involving \(\sigma\), \(\pi\), \(\cup\), \(\cap\), and \(-\) within the framework of categorical algebra. 

\smallskip

\textbf{Proof of Rule (7)}: We first prove that (i) 
\begin{equation} \label{equ:7left}
    \pi_L(Lim(Cat(R_1, R_2, f_1:R_1 \to R_2))) \subseteq Lim(Cat(\pi_{L_1}(R_1), \pi_{L_2}(R_2), f_2: \pi_{L_1}(R_1) \to \pi_{L_2}(R_2)) ) 
\end{equation} 

Given any element $r$ in $\pi_L(Lim(Cat(R_1, R_2, f_1:R_1 \to R_2)))$, $\exists r_1 \in R_1, r_2 \in R_2$ and $f_1(r_1) = r_2$, $r = \pi_L(r_1,r_2)$=$(\pi_{L_1}r_1, \pi_{L_2}r_2) = (\pi_{L_1}r_1, \pi_{L_2}(f_1(r_1))) = (\pi_{L_1}r_1, f_2(\pi_{L_1}(r_1)))$. Therefore, 

$r \in Lim(Cat(\pi_{L_1}(R_1), \pi_{L_2}(R_2), f_2: \pi_{L_1}(R_1) \to \pi_{L_2}(R_2))$.

We then prove that (ii)
\begin{equation} \label{equ:7right}
   Lim(Cat(\pi_{L_1}(R_1), \pi_{L_2}(R_2), f_2: \pi_{L_1}(R_1) \to \pi_{L_2}(R_2)) )   \subseteq \pi_L(Lim(Cat(R_1, R_2, f_1:R_1 \to R_2)))  
\end{equation} 

Given any element ($r_1,r_2$) in $Lim(Cat(\pi_{L_1}(R_1), \pi_{L_2}(R_2), f_2: \pi_{L_1}(R_1) \to \pi_{L_2}(R_2)))$, $f_2(r_1)=r_2$. Thus, $\exists$$r_1'$$\in R_1$, s.t. $(r_1,r_2) = (r_1,f_2(r_1)) = (\pi_{L_1}(r_1'),f_2(\pi_{L_1}(r_1'))) = (\pi_{L_1}(r_1'),\pi_{L_2}(f_1(r_1'))) = (\pi_{L}(r_1'),\pi_{L}(f_1(r_1')))$. Therefore, ($r_1,r_2$) is also in $\pi_L(Lim(Cat(R_1, R_2, f_1:R_1 \to R_2)))$.

Combining Formulas (\ref{equ:7left}) and (\ref{equ:7right}) yields the result, thereby proving the rule.

\smallskip

\textbf{Proof of Rule (8)}:  Similar to the proof for Rule (7), we first prove that (i) 
\begin{equation} \label{equ:left}
    g(Lim(Cat(S_1, S_2, f_1:S_1 \to S_2))) \subseteq Lim(Cat(g_1(S_1), g_2(S_2), f_2: g_1(S_1) \to g_2(S_2)) ) 
\end{equation} 

Consider any element $(x_1,x_2)$ in the set of $Lim(Cat(S_1, S_2, f_1:S_1 \to S_2))$, we have $g(x_1,x_2)=(g_1(x_1),  g_2(x_2))$, as $g= g_1 \otimes g_2$. Further,  $(g_1(x_1),  g_2(x_2)) = (g_1(x_1),  g_2(f_1(x_1))) = (g_1(x_1),  f_2(g_1(x_1))) $, which belongs to $Lim(Cat(g_1(S_1), g_2(S_2), f_2: g_1(S_1) \to g_2(S_2)) )$ as desired.

We then prove that (ii) \begin{equation} \label{equ:right} Lim(Cat(g_1(S_1), g_2(S_2), f_2: g_1(S_1) \to g_2(S_2)) \subseteq
    g(Lim(Cat(S_1, S_2, f_1:S_1 \to S_2)))   ) 
\end{equation}

Given any element $(y_1, y_2) \in Lim(Cat(g_1(S_1), g_2(S_2), f_2: g_1(S_1) \to g_2(S_2))) $ $\Rightarrow$ $\exists x_1 \in S_1$ and $\exists x_2 \in S_2$, s.t. $y_1 = g_1(x_1)$ and $y_2 = g_2(x_2)$ $\Rightarrow$ $(g_1(x_1), g_2(x_2)) =  (g_1(x_1), f_2(g_1(x_1))) = (g_1(x_1),  g_2(f_1(x_1))) = g(x_1,f_1(x_1))$ $\Rightarrow$ $(y_1, y_2) \in g(Lim(Cat(S_1, S_2, f_1:S_1 \to S_2))) $.

Combining Formulas (\ref{equ:left}) and (\ref{equ:right}) yields the result.

\smallskip

We now present more categorical transformation rules, which have the similar expression in relational algebra.

  1. \textbf{Cascade of $\sigma$}. A conjunctive selection condition can be broken up into a cascade of individual $\sigma$ operations. \[ \sigma_{C_1 \wedge \ldots \wedge C_n}(R)  \equiv \sigma_{C_1}   ( \ldots (\sigma_{C_n}(R))    \]

   2. \textbf{Commutative of $\sigma$}. The $\sigma$ operation is commutative. \[ \sigma_{C_1} (\sigma_{C_2}(R))  \equiv \sigma_{C_2} (\sigma_{C_1}(R))    \]

  3. \textbf{Cascade of $\pi$}. In a cascade of $\pi$ operations, all but the last one can be ignored. \[ {\pi_{List_1}}   ( \ldots ({\pi_{List_n}}(R)))   \equiv {\pi_{List_1}}(R)    \]

4. \textbf{Pushing $\sigma$ in conjunction with set operators}. \[ \sigma_C (R \cup S)   \equiv \sigma_C(R) \cup \sigma_C(S)    \] \[ \sigma_C (R \cap S)   \equiv \sigma_C(R) \cap \sigma_C(S)    \] 

\end{document}